\documentclass{article}

\usepackage{lineno,hyperref}
\usepackage{epsfig}
\usepackage{graphicx}
\usepackage{color}
\usepackage{subfigure}
\usepackage{graphics}
\usepackage{dcolumn}
\usepackage{bm}
\usepackage{multirow}
\usepackage{longtable}
\usepackage[ruled, lined]{algorithm2e}        
\usepackage{setspace}
\usepackage{amsthm}
\usepackage{ragged2e}
\usepackage{mathabx}
\usepackage{textcomp}
\usepackage[table]{xcolor}
\usepackage{tabularx}
\usepackage{changepage}

\newtheorem{theorem}{Theorem}
\newtheorem{prop}{Proposition}

\SetAlFnt{\small}
\SetAlCapFnt{\small}
\SetAlCapNameFnt{\small}
\SetAlCapHSkip{0pt}
\IncMargin{-\parindent}

\begin{document}

\textbf{
\large {\centerline{
A Graph-Based Approach to Analyze Flux-Balanced Pathways}}}
\large{\textbf{\centerline{in Metabolic Networks}}}
\vspace{1em}

\small{
\centerline{Mona Arabzadeh$^1$, Morteza Saheb Zamani$^1$, Mehdi Sedighi$^1$,}
\centerline{and Sayed-Amir Marashi$^2$}}
\vspace{1em}
\footnotesize{
\centerline{$^1$Department of Computer Engineering and Information Technology,}
\centerline{Amirkabir University of Technology, Tehran, Iran and}
\vspace{0.2em}
\centerline{$^2$Department of Biotechnology, College of Science, University of Tehran, Tehran, Iran.}
\vspace{1em}
\centerline{\{m.arabzadeh,szamani,msedighi\}@aut.ac.ir, marashi@ut.ac.ir}
}

\begin{abstract}
An Elementary Flux Mode (EFM) is a pathway with minimum set of reactions that are functional in steady-state constrained space.
Due to the high computational complexity of calculating EFMs, different approaches have been proposed to find these flux-balanced pathways.
In this paper, an approach to find a subset of EFMs is proposed based on a graph data model. The given metabolic network is mapped to the graph model and decisions for reaction inclusion can be made based on metabolites and their associated reactions. This notion makes the approach more convenient to categorize the output pathways.
Implications of the proposed method on metabolic networks are discussed.
\end{abstract}



\normalsize
\section{Introduction}
Metabolic network models are among the well-studied models in biotechnology. The reconstruction of these networks is possible by collecting the gene-protein-reaction information from related genomic data and literature~\cite{henry2010high}. It is important to explore biologically relevant pathways in metabolic networks. Forcing constraints to a reconstructed biochemical network results in the definition of achievable cellular functions~\cite{price2004genome}. Mathematical representation of constraints are as \emph{flux-balance} constraints (e.g., conservation of flux) which means the network should be at the steady-state condition, and \emph{flux bounds} which limit numerical ranges of network parameters and coefficients such as the minimum and maximum range of fluxes for each reaction. Flux Balance Analysis (FBA), a method to predict the optimal growth rate of a a certain species when grown on particular set of metabolites~\cite{orth2010flux}, and Elementary Flux Mode (EFM) analysis, an approach to decompose a network to minimal functional pathways~\cite{klipp2008systems}, are among the approaches used in constraint-based analysis of metabolic networks.

EFMs~\cite{schuster1994elementary,schuster2000general} have been used in several biological applications such as
bioengineering \cite{schuster2002use},
phenotypic characterization \cite{radhakrishnan2010phenotypic},
drug target prediction \cite{parvatham2013drug} and
strain design \cite{machado2015co}.
The incorporation of kinetic analysis into EFMs enables a more complete description of cellular functions for which kinetics play a dominant role~\cite{papin2003metabolic}.

Several methods were introduced for finding EFMs based on double-description method \cite{schuster1994elementary}. Double-description is a technique to enumerate all extreme rays of a polyhedral cone. An improved approach to the primary method was introduced in which the null-space of the stoichiometric matrix is used instead of the matrix itself to generate EFM candidates \cite{wagner2004nullspace,urbanczik2005improved,quek2014depth}.
Various effective computational approaches have been proposed to speed-up previous methods for computing EFMs~\cite{gagneur2004computation,terzer2008large}. Some of these approaches led to the development of computational tools such as \emph{Metatool} \cite{von2006metatool} and \emph{EFMtool} \cite{terzer2008large}.
Besides, methods based on linear programming have been proposed that explore a set of EFMs with specific properties, such as $K$-shortest EFMs~\cite{de2009computing}, or EFMs with a given set of target reactions \cite{david2014computing}.
A new set of methods based on graph-theory, \cite{cespedes2015new,ullah2016gefm} tries to overcome the scalability problem of the double-description-based techniques.

According to the complexity of computing EFMs as discussed in~\cite{acuna2009modes,acuna2010note}, an additional challenge of extracting biological properties from large set of EFMs~\cite{papin2003metabolic} has been raised. As stated in~\cite{acuna2010note}, the complexity of enumerating all EFMs still remains unknown and the computational complexity of enumerating EFMs containing a specific reaction is hard.
Even for networks with the same number of reactions and metabolites,
knowing the number of EFMs in one network cannot necessarily help in finding the number of EFMs in the other one, as they have different connectivities and topologies.
This fact emphasizes the inherent structural information reflected by EFMs~\cite{klamt2002combinatorial} and is a motivation to extract a set of biologically meaningful EFMs according to a biological aspect such as motifs~\cite{peres2011acom} and thermodynamics~\cite{gerstl2015metabolomics}.

The main focus of this paper is to introduce a data model based on the AND/OR graph and to propose an approach to find flux-balanced pathways according to pathway topology and reaction stoichiometries. It is shown that the computed flux-balanced pathways are a biologically relevant subset of EFMs which include external input and output metabolites. In other words, the subset comprises all EFMs connecting input and output external metabolites.
Besides, based on the introduced graph data model, an upper-bound for the complexity of exploring EFMs containing external metabolites is calculated.
Using the topology of the network gives us the opportunity to make decisions according to the metabolite/reaction aspects to preserve or eliminate a particular reaction or metabolite in a certain pathway. Introducing a model to consider both \emph{topology} and \emph{stoichiometry} of a metabolic network for network analysis may lead to better biological decisions and output categorization.
The graph structure makes it easier to set-up rules for pathways and to explore the intended solution space via rules. Finally, our method can potentially be implemented on hardware (through a system design approach) more conveniently.

The rest of the text is organized as follows. Some required concepts and preliminaries are provided in Section~2. Before explaining the proposed approach to find the elementary flux modes in a given metabolic network (Section~\ref{sec:GBEFM}), we first describe the modified AND/OR graph model and its properties in Section~\ref{sec:DataModel}. Finally, Section~\ref{sec:res} is devoted to results and discussions and Section~\ref{sec:conc} concludes the paper.


\section{Preliminaries}
In this section, some basic concepts along with the formal definition of elementary flux modes are provided.
\subsection{Metabolic Networks}
Metabolic networks model the metabolism of living cells in terms of a set of biochemical reactions.
Definitions in this section are derived from~\cite{zanghellini2013elementary}. The biochemical reactions can be irreversible, i.e., the reaction can be active only in one direction, or reversible, i.e., the reaction can be active in both directions. The contributing metabolites in a reaction can be either substrates or products. Substrates are consumed and products are produced during the operation of a reaction.
The topology of a metabolic network is characterized by its $m \times n$ stoichiometric matrix, $S$, where $m$ and $n$ correspond to the number of metabolites and reactions, respectively. In this paper, external metabolites are not included in $S$. The value $S_{ij}$ represents the stoichiometric coefficient of the metabolite $i$ in the reaction $j$. $S_{ij}$ is positive/negative if the metabolite $i$ is produced/consumed. If this coefficient is zero it means that the metabolite $i$ does not contribute to the reaction $j$.
The network is considered in the steady-state if for each internal metabolite, the rates of consumption and production are equal. The reactions connected to the external metabolites are called \emph{Boundary} reactions.


\subsection{Elementary Flux Mode Definition}
A flux vector $\textbf{v} \neq \textbf{0}$ and $\textbf{v} \in {R}^n$ is considered an EFM if it meets the following conditions:
\begin{itemize}
  \item $v_{i} \geq 0$ for all $i \in$ irreversible reactions (thermodynamic constraint).
  \item $\textbf{S}.\textbf{v} = \textbf{0}$ (steady-state condition).
  \item There is no $\textbf{v}^{\prime} \in {R}^n$ with $\textbf{supp}(\textbf{v}^{\prime}) \subset \textbf{supp}(\textbf{v})$, where support of a mode is defined as $\textbf{supp}(\textbf{v})= \{i|v_i \neq 0\}$, (minimality limitation).
\end{itemize}
\section{Methods}
In this section, the proposed graph-based elementary flux mode analysis is discussed in detail. To do so, the data structure used in our approach is defined, and then, the proposed approach is introduced accordingly.
\subsection{Data Model}\label{sec:DataModel}
The proposed data model is based on the conventional AND/OR graph in computer science. However, we provide a different definition with additional features to make the model appropriate for our proposed algorithm. In our model, the coefficients of the metabolites in reactions are embedded in the graph structure as attributes of each node (explained in Definition 3, see below). A $pathway$ is defined as a set of reactions with their associated fluxes. The flux of each reaction in a certain pathway is also embedded in the graph.\\

\noindent
\textbf{Definition 1.} The Modified Graph for representing a metabolic network, denoted as $\mathbf{MG}$, is defined as a set of \emph{Nodes}, $\mathcal{N}_i$, i.e., $\mathbf{MG}$ = $\{\mathcal{N}_i|0\leq i\leq {M-1}\}$, where $M$ is the number of internal metabolites in the metabolic network.

\noindent
\textbf{Definition 2.} Each node $\mathcal{N}_i$ in $\mathbf{MG}$ is a 3-tuple $\mathcal{N}$ = ($i$, $I$, $O$) where
\begin{itemize}
  \item $i$ is the tag of a metabolite,
  \item $I$ is an array of input reactions that produce the metabolite and
  \item $O$ is an array of output reactions that consume the metabolite.
\end{itemize}

\noindent
\textbf{Definition 3.} Each $I$/$O$ in Definition 2 contains the following data:
\begin{itemize}
  \item The reaction $j$, $0\leq j\leq r-1$, where $r$ is the number of reactions consuming/producing the metabolite $i$,
  \item $I_{M_j}$/$O_{M_j}$, an array of the metabolites consumed/produced by reaction $j$. In other words, $I_{M_j}$/$O_{M_j}$=$\{m_{kj}|0\leq k\leq m-1\}$, where $m$ is the number of consumed/produced metabolites by the reaction $j$,
  \item $I_{{\mathord{\buildrel{\lower3pt\hbox{$\scriptscriptstyle\frown$}}\over M}}_j}$/$O_{{\mathord{\buildrel{\lower3pt\hbox{$\scriptscriptstyle\frown$}}\over M}}_j}$, an array of the metabolites produced/consumed by reaction $j$. In other words, $I_{{\mathord{\buildrel{\lower3pt\hbox{$\scriptscriptstyle\frown$}}\over M}}_j}$/$O_{{\mathord{\buildrel{\lower3pt\hbox{$\scriptscriptstyle\frown$}}\over M}}_j}$=$\{m_{kj}|0\leq k\leq m-1\}$, where $m$ is the number of produced/consumed metabolites by the reaction $j$ excluding the metabolite $i$ itself,
  \item The direction of the reaction for reversible reactions,
  \item $I_{c_{ij}}$/$O_{c_{ij}}$, the coefficient of the reaction $j$ for the produced/consumed metabolite $i$ in the stoichiometric matrix $S$ and
  \item $I_{f_{ijp}}$/$O_{f_{ijp}}$, the flux of the input/output reaction $j$ for a certain pathway $p$.
\end{itemize}

\noindent
A general configuration of a node $\mathcal{N}_i$ in $\mathbf{MG}$ is shown in Figure~\ref{fig:AndOR}. The incoming arcs (i.e., $I$ in Definition 2) in each $\mathcal{N}_i$ produce the metabolite $i$ and the outgoing arcs (i.e., $O$ in Definition 2) consume it. Therefore, consumed metabolites, $m_i$, and produced metabolites, $m^{\prime}_{i}$, contributing to reaction $j$ as shown in Equation~\ref{eq:3}, are as arcs between $\mathcal{N}_i$ nodes each associated to one metabolite. $I$ arcs and $O$ arcs in $\mathcal{N}_i$ nodes are related to each other by reaction tags in Definition 3.

\begin{equation}\label{eq:3}
r_j: m_1 + m_2 + ... + m_i + ...\rightleftharpoons m^{\prime}_{1} + m^{\prime}_{2} + ... +m^{\prime}_{i}+ ...
\end{equation}

\begin{figure}[!tpb]
\centering
\includegraphics[scale=0.4]{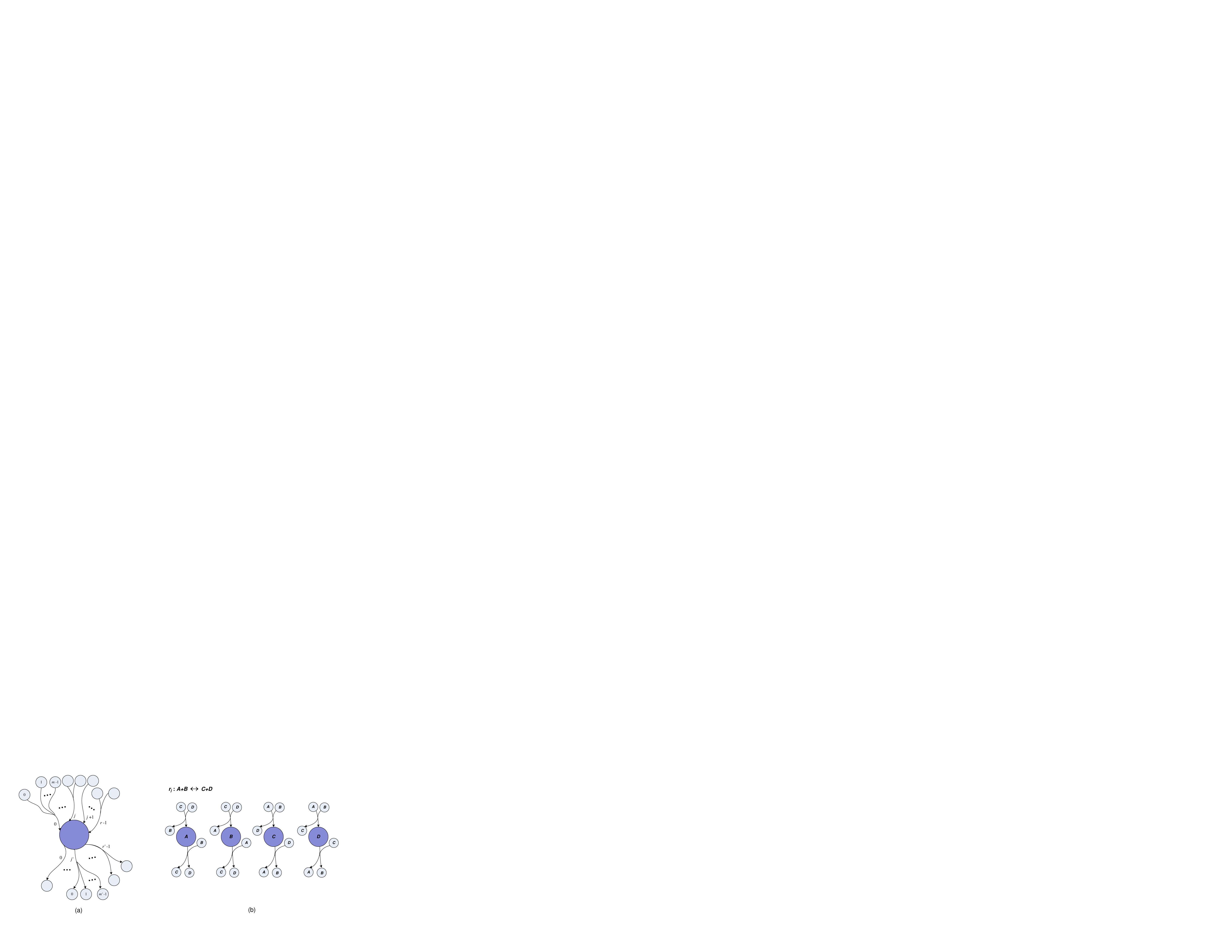}
\caption{(a) A general configuration of a node $\mathcal{N}_i$ in a modified AND/OR graph as defined in Definitions 1, 2 and 3. The node has $r$ input reactions and $r^{\prime}$ output reactions. The maximum number of AND metabolites for an input reaction $j$ and for an output reaction $j^{\prime}$ is considered as $m$ and $m^{\prime}$, respectively. (b) The way the reversible reaction $r_j$ is stored in an MG model: MG = \{A,B,C,D\}; since the reaction is reversible, it is added to both $I$ and $O$ array of each node. For example for node $A$, $I_{M_j} = \{C,D\}$ and $I_{{\mathord{\buildrel{\lower3pt\hbox{$\scriptscriptstyle\frown$}}\over M}}_j}=\{B\}$;
$O_{M_j}=\{C,D\}$ and $O_{{\mathord{\buildrel{\lower3pt\hbox{$\scriptscriptstyle\frown$}}\over M}}_j}=\{B\}$.}
\label{fig:AndOR}
\end{figure}

\subsection{GB-EFM: Graph-Based EFM Analysis Algorithm} \label{sec:GBEFM}
The proposed algorithm for finding EFMs of a metabolic network is based on the following facts:
\begin{enumerate}
  \item Starting from boundary metabolites, each incoming flux in a pathway produce some metabolite(s), and each produced metabolite should be consumed and the flux of the contributed reactions can be obtained according to the stoichiometric coefficients in order to make sure that the constructed pathway is in steady-state.
  \item To guarantee the elementarity of the produced pathways, two rules are followed: (1) Only one output from the node $\mathcal{N}_i$ should be considered when constructing a pathway and (2) multi-path condition should be checked in a node which has more than one input and one output. The multi-path condition occurs in a node when independent pathways share part of a pathway while traversing the graph.
  \item When one metabolite is consumed by a reaction, the presence of its AND-related metabolites are required. For example, when metabolites $A$ and $B$ are AND-related, if $A$ is consumed during the reaction $A$+$B$$\rightarrow$$C$+$D$, then $B$ should also be present which means $A$ and $B$ should be consumed simultaneously. Therefore, the algorithm is designed to have a forward/backward flow for the produced/consumed metabolites. In this example, the metabolites $A$ and $B$ are consumed while $C$ and $D$ are produced.
\end{enumerate}

\begin{algorithm}[!tpb]
\DontPrintSemicolon
\SetAlgoLined
\small
\textbf{          }
\KwIn{$S$: The stoichiometric matrix of the internal metabolites in a metabolic network and an array $V$ representing reversible reactions}
\textbf{          }
\KwOut{The equivalent $\mathbf{MG}$ of $S$}
\textbf{ }\\
\For{all metabolites $i$, $0\leq i\leq M-1$, the rows of $S$}
{
\For{all reactions $j$, $0\leq j\leq R-1$, the columns of $S$}
    {
    Create node $\mathcal{N}_i$.\\
    \For{all irreversible reactions}
    {\uIf {$S_{ij}> 0$} {add an input to $\mathcal{N}_i$ with the reaction tag $j$.}
    \ElseIf {$S_{ij}<0$} {add an output to $\mathcal{N}_i$ with the reaction tag $j$.}}
     \For{all reversible reactions $v_i \in V$}
     {Add both an input and an output to $\mathcal{N}_i$ with the reaction tag $j$.$^*$}
    }
}

\scriptsize
\justify
*The model demonstrates a hierarchy. Each node has a set of reactions (tagged as input or output) and each reaction itself has an array of input and-related metabolites and an array of output and-related metabolites. When we add a reaction to both input array and output array of a metabolite, each input/output becomes an object with the same tag $j$, that is the reaction name. This affirms how AND/OR information is preserved in the model.

\caption{Construction of Modified AND/OR Graph.}
\label{alg:one}
\end{algorithm}

\noindent
The proposed graph-based EFM analysis algorithm is composed of the following five steps:

\noindent
\textbf{Step 0. Construction of $\mathbf{MG}$.} The graph $\mathbf{MG}$ is constructed based on the stoichiometric matrix as stated in Algorithm 1.\\


\noindent
\textbf{Step 1. Finding independent pathways with dependent nodes.} In this step, starting from the metabolites connected to the boundary reactions, all possible pathways are constructed by traversing the internal metabolites.
Proposition~\ref{prop:step1} shows that the pathways are constructed with minimum number of possible reactions using Definition 4 and Definition 5. The semi-minimality occurs since different pathways are constructed from different outputs and each node is visited multiple times only if required.

\noindent
\textbf{Definition 4. Forward/Backward Flow.}
The \emph{forward} flow is applied for the output metabolites of the selected reaction and the \emph{backward} flow is applied for the and-related metabolites of that reaction.

\noindent
\textbf{Definition 5. Primary and Secondary Reactions.}
On each pathway, each node has a \emph{primary} input and a \emph{primary} output which directs the flow of the pathway and tags the metabolite as \emph{visited}. When an edge (i.e., a reaction) enters an already visited node in that pathway in a forward/backward flow as an input/output, that reaction is tagged as a \emph{secondary} input/output of that node.

\begin{prop}\label{prop:step1}
Starting from metabolites connected to the boundary reactions, $r$ pathways are constructed from each node where $r$ is the number of output/input reactions of that node in the forward/backward flow. The pathways are \emph{semi-minimal} based on the following statements:

\begin{itemize}
  \item each output/input reaction in forward/backward flow is traversed once and labeled as \emph{primary}, and
  \item when an edge enters an already visited node, the pathway is considered as closed in the node and the edge is labeled as \emph{secondary}.
\end{itemize}
\end{prop}

In the forward flow of a node, for each output reaction, an independent pathway is constructed. Each pathway is traversed independently of the other pathways. In the backward flow of a node, for each input reaction an independent pathway is constructed.
In Figure~\ref{fig:step1} the forward/backward flow of each node $\mathcal{N}_i$ is shown in an illustrative example. The information of each pathway is saved independently.

\begin{figure}[!tpb]
\centering
\includegraphics[scale=0.6]{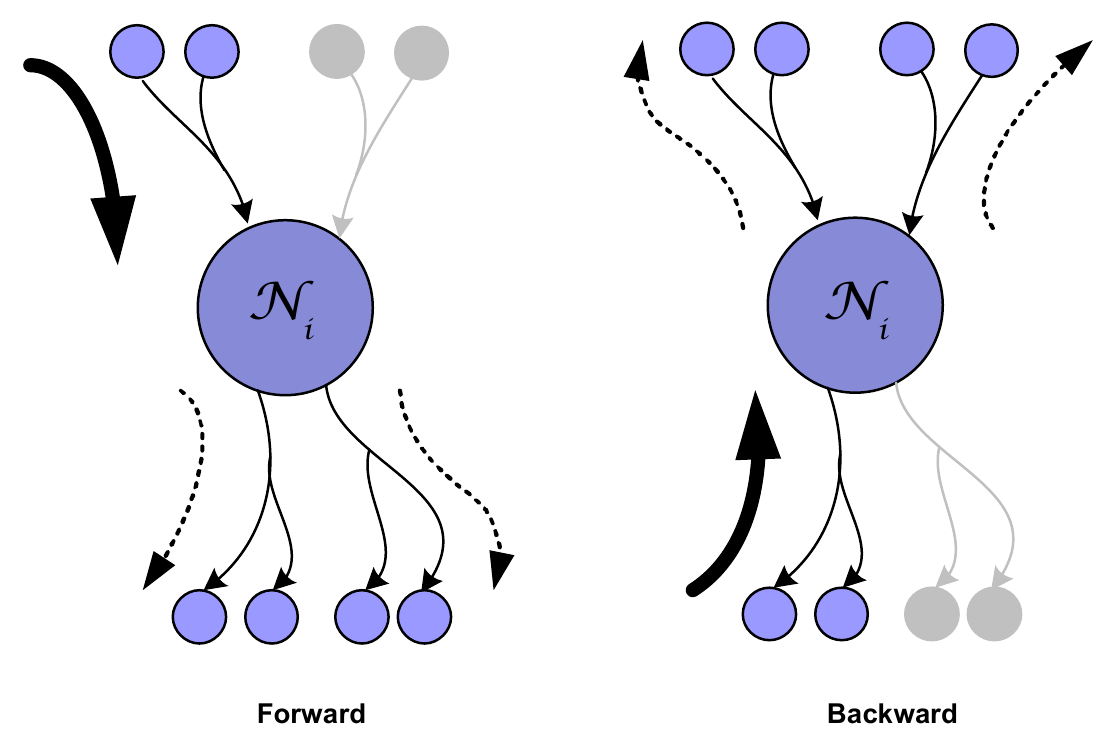}
\caption{Constructing possible pathways from outputs/inputs in forward/backward flow in Step~1. In the forward/backward flow, from the primary input/output illustrated by the thick arrow, two independent pathways are constructed along each output/input reaction illustrated by dashed arrows.}
\label{fig:step1}
\end{figure}


\noindent
\textbf{Step 2. Eliminate multi-way pathways.}
By using Proposition~\ref{prop:step2} and eliminating multi-way pathways (defined in Definition 7) which include nodes with multi-path condition (defined in Definition 6),
the set of semi-minimal paths is reduced to the set of minimal pathways.
An example of this is shown in Figure~\ref{fig:multiPath}. 

\noindent
\textbf{Definition 6. Multi-Path Condition.} This is a condition that occurs in a node with more than one input and more than one output and two independent pathways from different origins share this node.

\noindent
\textbf{Definition 7. Multi-Way Pathways.} Semi-minimal pathways that contain reactions from two independent pathways are referred to as multi-way pathways.

\begin{prop}\label{prop:step2}
Since different pathways are constructed from different outputs of a node and there are both forward and backward flows, depending on the order of traversing the nodes,
there exist multi-way pathways containing reactions from two independent pathways.
The pathway is minimal if it is not a multi-way pathway. In other words, it has no node with multi-path condition applied to it.
\end{prop}

To solve the multi-path condition, for each produced pathway in Step 1, we traverse the subgraph of that pathway and tag each reaction with its associated metabolites. Since the structure of the pathway is known, one can label the outgoing/ingoing reactions in forward/backward flow of a pathway.

Then, in a pathway, for each node with more than one input and more than one output reaction, if one can find two routes crossing the node and the origins of the routes are different, the pathway should be discarded.
This is checked according to the labels, i.e., $[node.pathNum]$, which has been assigned to each reaction. $pathNum$ is a variable which counts the number of pathways generated by a node.
In other words, this pathway can be functional by either of the found routes in the node.

\begin{figure}[!tpb]
\centerline{
\includegraphics[scale=1.5]{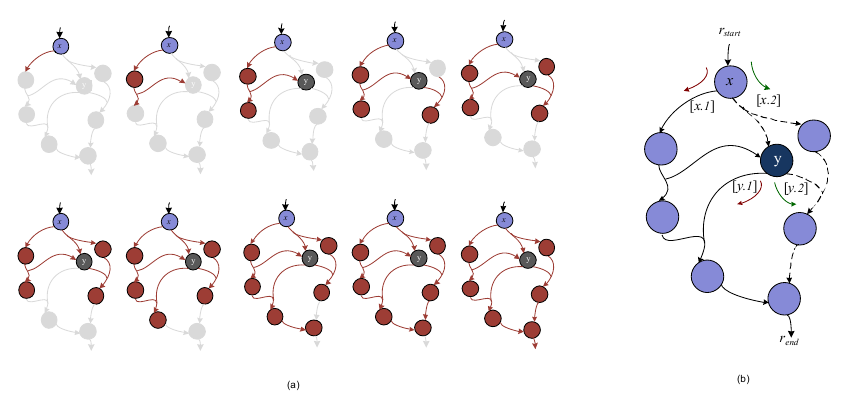}}
\caption{An example of the multi-path condition and a multi-way path. (a) This figure shows how the order of traversing the nodes may cause a multi-path condition in node $y$. (b) In general, starting from node $x$, two paths are constructed and their reactions are labeled as $x.1$ and $x.2$. After arriving at node $y$ from each of the paths, two other paths are constructed and their reactions are labeled as $y.1$ and $y.2$. Altogether four different paths are obtained, including the reaction pairs ($[x.1]$,$[y.1]$), ($[x.1]$,$[y.2]$), ($[x.2]$,$[y.1]$) and ($[x.2]$,$[y.2]$). As can be seen in the graph, the paths with reactions labeled with 1 are converged to each other and the paths which their reactions labeled with 2 are converged to each other as well. Therefore, two of the four paths passing ($[x.1]$,$[y.2]$), ($[x.2]$,$[y.1]$) in fact belong to the same metabolic pathway but the two other paths, path 1 shown by solid lines and path 2 shown by dashed lines, are independent. The starting and ending reactions belong to both paths. Step~2 identifies this by labeling the outputs and finding the multi-path nodes to eliminate multi-way paths which are a combination of independent paths.}
\label{fig:multiPath}
\end{figure}

\noindent
\textbf{Step 3. Adjusting reaction fluxes.}
In this step, every node $i$ in a pathway $p$ is examined and the fluxes of the primary input and the primary output of that node are calculated using the following equations:

\begin{equation}\label{eq:1}
 \begin{array}{l}
 O_{f_{ijp} }  = {\rm   }\frac{{I_{f_{ijp} } I_{c_{ij} } }}{{O_{c_{ij} } }},\,\,\,\,\,\,\,\,\,\,\,\,\,\,\,\,Forward \\\\
 I_{f_{ijp} }  = {\rm   }\frac{{O_{f_{ijp} } O_{c_{ij} } }}{{I_{c_{ij} } }},\,\,\,\,\,\,\,\,\,\,\,\,\,\,\,Backward \\
 \end{array}
\end{equation}

\noindent
In Equation~\ref{eq:1}, $j$ refers to the primary input/output reaction. Proposition~\ref{prop:step3} explains how all reactions in a pathway $p$ get a flux.

\begin{prop}\label{prop:step3}
Starting from a boundary reaction with a certain flux, all nodes are being traversed and all reactions in a pathway get a flux using Equation~\ref{eq:1}, as each node has a primary input and a primary output to be used in this equation. Equation~\ref{eq:1} states that the amount of the incoming and outgoing fluxes in $\mathcal{N}_i$ should be equal.
\end{prop}

\noindent
\textbf{Step 4. Balancing flux.}
After Step 3, the consumption/production balance of the nodes with secondary reactions may be disturbed. To fix this, we go through these nodes and find a way forward/backward to output/input reactions to recalculate the extra production/consumption of the unbalanced nodes. Using Proposition~\ref{prop:step4}, the pathways with no balance rate are discarded.
If $\sum {I_{c_{ij} }I_{f_{ijp} }  }  - \sum {O_{c_{ij} }O_{f_{ijp} }  }  = 0$, the balance condition is accepted for this node. Otherwise, on this pathway, the extra flux is added to the next reactions on the pathway, towards the boundary reactions using Equation~\ref{eq:2}. In this case, the difference between the input rates and the output rates is equal to $I_{c_{ie}} I_{f_{iep} } $ or $O_{c_{ie}}O_{f_{iep} } $ where $e$ implies the reaction which causes the extra production/consumption of the $\mathcal{N}_i$. The old/new value of the flux is shown by $o$/$n$ labels in Equation~\ref{eq:2} and $j$ refers to the primary input/output reaction.

\begin{equation}\label{eq:2}
\begin{array}{l}
 I_{f(n)_{ijp} }  = \frac{{I_{f(o)_{ijp} } I_{c_{ij} }  - I_{f_{iep} } I_{c_{ie} } }}{{I_{c_{ij} } }},\, \\\\
 I_{f(n)_{ijp} }  = \frac{{I_{f(o)_{ijp} } I_{c_{ij} }  + O_{f_{iep} } I_{c_{ie} } }}{{I_{c_{ij} } }},\, \\\\
 O_{f(n)_{ijp} }  = \frac{{O_{f(o)_{ijp} } O_{c_{ij} }  + O_{f_{iep} } I_{c_{ie} } }}{{O_{c_{ij} } }}, \\\\
 O_{f(n)_{ijp} }  = \frac{{O_{f(o)_{ijp} } O_{c_{ij} }  - I_{f_{iep} } I_{c_{ie} } }}{{O_{c_{ij} } }} \\
 \end{array}
\end{equation}

\begin{prop}\label{prop:step4}
Flux dependencies can prevent the rate of the production/consumption of an internal metabolite in a given topology of a pathway from being zero. Equation~\ref{eq:2} is used to calculate an update for the consumption/production rate of the node. In the case of
\begin{itemize}
  \item $I_{f(n)_{ijp} }=0$\,\,\,\,\,\,\,\,\,\,\,\, or \,\,\,\,\,\,\,\,\,\,\,\,$O_{f(n)_{ijp} } = 0$ or
  \item a repetitive loop over a node (i.e., getting back to a node from the same reaction multiple times), while trying to find a way out in a subgraph of a given pathway,
\end{itemize}
the pathway is discarded.
\end{prop}

\begin{figure}[!tpb]
\centering
\includegraphics[scale=0.55]{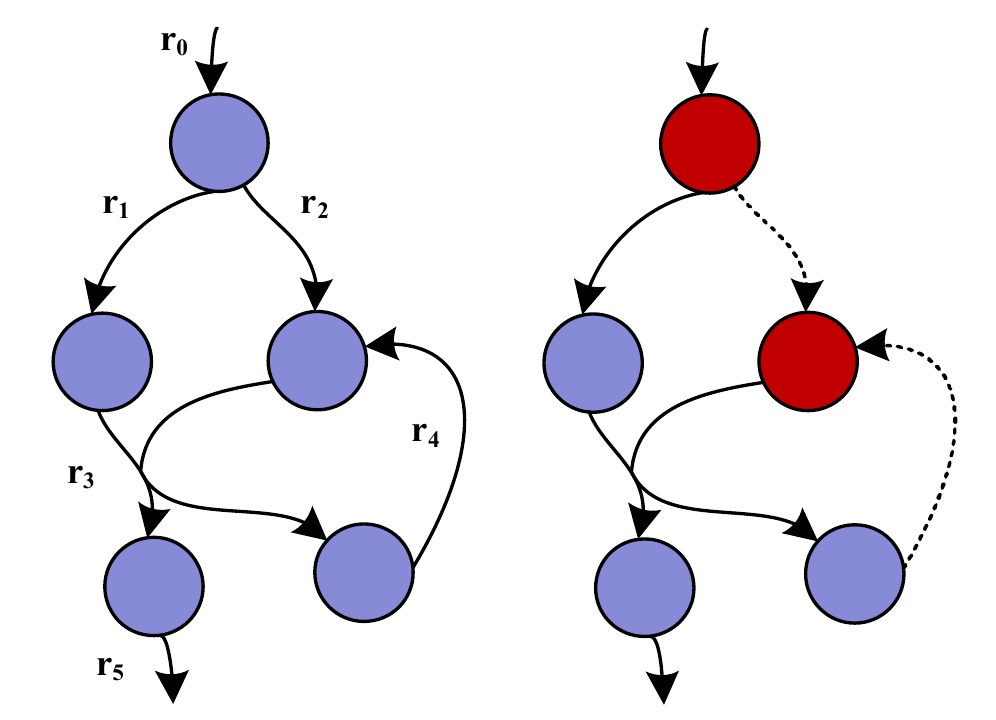}
\caption{An illustration of Steps~3 and 4 on a sample graph. The pairs ($r_0$,$r_1$), ($r_1$,$r_3$), ($r_3$,$r_5$), ($r_3$,$r_2$), ($r_3$,$r_4$) act as primary reactions (input, output) for the five nodes in Step~3. The two nodes with dashed secondary reactions update production/consumption of metabolites for $r_2$/$r_4$ in Step~4.}
\label{fig:step3-4}
\end{figure}

An illustration of Steps 3 and 4 on a sample graph is shown in Figure~\ref{fig:step3-4}.
\begin{algorithm*}
\label{alg:two}
\caption{Graph-Based EFM Analysis Algorithm.}
\textbf{         }\\
\DontPrintSemicolon
\SetAlgoLined
\scriptsize
\KwIn{The $\mathbf{MG}$ graph of a metabolic network derived from its given stoichiometric matrix.}
\KwOut{The set of elementary flux modes of the given network from input metabolites to output metabolites.}
\textbf{         }\\
Store all boundary metabolites of the network with their associated reactions in the queue $\mathcal{Q}$.\\
\For {each entry $\mathcal{Q}_i$ in the queue, $0\leq i\leq size(\mathcal{Q})- 1$}
{
    \For {each $O_j$/$I_j$, in forward/backward flow}
  {
      Start a new path for each reaction $j$ as a primary output/input. If the reaction is reversible, disable the reverse input/output direct.\\
      \uIf {$\mathcal{N}_i \notin \mathcal{P} $}
        {Add $\mathcal{N}_i$ to $\mathcal{P}$ accordingly for each path.}
      \Else
        {Add $j$ to $\mathcal{N}_i$ as a secondary output/input.}
    Push back $O_{M_i}$/$I_{M_i}$ to $\mathcal{Q}$ to be considered in the forward path of the algorithm.\\
    Push back $O_{{{\mathord{\buildrel{\lower3pt\hbox{$\scriptscriptstyle\frown$}}\over M} }}_i}$/$I_{{{\mathord{\buildrel{\lower3pt\hbox{$\scriptscriptstyle\frown$}}\over M} }}_i}$ to $\mathcal{Q}$ to be considered in the backward path of the algorithm.\\
    Repeat until $\mathcal{Q}$=$\emptyset$.}}
\textbf{         }\\
\For {all $\mathcal{P}_k$, $0\leq k \leq P-1$}
{
    \For {all $\mathcal{N}_i$ in $\mathcal{P}_k$}
    {
        \uIf { j > 1} {
            \For {all $O_j$/$I_j$, $0\leq j \leq r-1$ in forward/backward flow}
            {
                Add label $N_{ij}$ to $O_j$/$I_j$.
            }
        }
        \Else {
                Pass the label from input/output to output/input in forward/backward flow.
        }
    }
    \For {all $\mathcal{N}_i$ in $\mathcal{P}_k$ with more than two inputs and two outputs}
    {
        \If{ $\mathcal{N}_{x} \in \bigcup\limits_M{\mathcal{N}_{m}}$ where $M$ contains all visited metabolites associated with the reaction $j$}
        { \If {a cross match for $\mathcal{N}_{xy}$ and $\mathcal{N}_{xz}$ in the inputs is found as well as a match in the outputs}
        {
            Discard $\mathcal{P}_k$.
        }
        }
    }
}
\textbf{         }\\
\For {all remaining $\mathcal{P}_k$, $0\leq k \leq P-1$}
{
    \For {all $\mathcal{N}_i$ in $\mathcal{P}_k$}
    {
    Apply Eq. 1.
    }
}
\textbf{         }\\
\For {all remaining $\mathcal{P}_k$, $0\leq k \leq P-1$}
{
 \For {all $\mathcal{N}_i$ in $\mathcal{P}_k$ with Secondary inputs/outputs}
    {
        Calculate $\sum\limits_j {I_{f_{iek} } I_{c_{ie} } }  = \sum\limits_j {I_{f_{ijk} } I_{c_{ij} } }  - \sum\limits_j {O_{f_{ijk} } O_{c_{ij} } }$.\\
        \If {$\sum\limits_j {I_{f_{iek} } I_{c_{ie} } }\neq 0$}
        {
               Use Eq. 2 to pass the flux to the primary input/output reaction of $\mathcal{N}_i$.\\
               Repeat until reaching the boundary reactions.
        }
    }

}
\end{algorithm*}

The proposed GB-EFM analysis algorithm is presented in Algorithm 2 and is summarized in Theorem 1.

\begin{theorem}\label{th:main}
Starting from boundary nodes in a metabolic network, GB-EFM analysis algorithm produces all minimal flux-balanced pathways connecting input and output metabolites with the following properties:
\begin{itemize}
  \item${R_k }  \not\subset {R_l }$,\,\,\,\,for all\,\,\,\,$\mathcal{P}_k$ and $\mathcal{P}_l$,\,\,\,\,\,\,\,$k,l \in [0,P - 1]$\,\,\,\,\, $k \neq l$. \\
  \item $\sum\limits_j {I_{f_{ijk} } I_{c_{ij} } }  = \sum\limits_j {O_{f_{ijk} } O_{c_{ij} } }$,\,\,\, for all nodes\,\,\,\, $\mathcal{N}_i$,\,\,\,\,\,\,$0 \le i \le M-1$.
\end{itemize}
where $\mathcal{P}_k$ refers to a pathway $0\leq k \leq P-1$ and $R_k$ is the set of contributing reactions in the pathway. The above properties indicate that the produced pathways are elementary flux modes.
\end{theorem}

\begin{proof}
Step~0 provides the required interface from the stoichiometric matrix. By using Propositions~1 and~2, semi-minimal and then minimal pathways are produced, respectively. According to Propositions~3 and~4, steady-state condition is checked for the given pathways and fluxes are calculated. Non-qualified pathways are discarded in Steps~2 and~4 based on Propositions~\ref{prop:step2} and \ref{prop:step4}. The remaining pathways have minimum functional reactions with balanced fluxes as elementary flux modes. Therefore, the produced pathways are a subset of EFMs connecting input and output external metabolites.
\end{proof}

A complete framework of Algorithm 2 is illustrated in Figure~\ref{fig:fig}.

\begin{figure}[!tpb]
\centerline{
\includegraphics[scale=0.4]{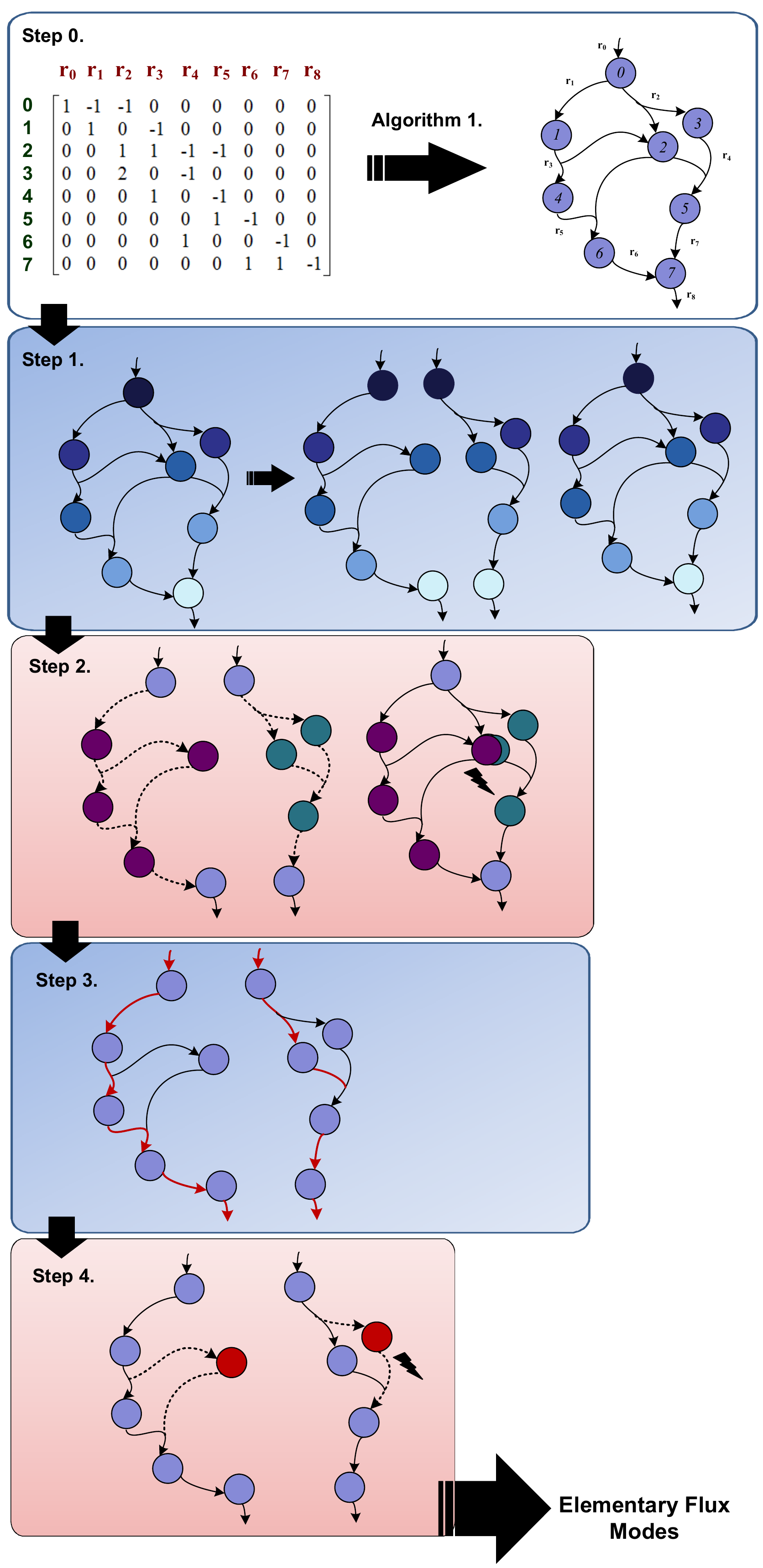}}
\caption{An example illustrating Steps~0 to 4 on a simple network. In Step 1, three different paths are constructed from a given $\mathbf{MG}$ produced in Step 0. The paths are divided as we go further in the graph, as shown by dark-colored nodes changing to light-colored nodes.
In Step 2, for the first two paths from left, there is no node with more than one input and one output. Therefore, multi-path condition is not applied.  However, the third path is constructed from the first two paths. The internal nodes of each of these two paths are shown by different colors. Multi-path condition is checked in their common node shown by overlapped circles. As a result, this path is discarded.
In Step 3, for each path, all nodes are traversed and primary inputs and primary outputs, illustrated by red solid lines, get a flux.
In Step 4, nodes with secondary input(s)/output(s) are traversed again to pass the extra flux to the boundary reactions. Secondary input(s)/output(s) are illustrated by dotted lines and their associated nodes are illustrated by red circles. According to the coefficients, as shown in the given stoichiometric matrix embedded in the graph, the second path cannot get a stable flux value for all reactions and is discarded.
The final output after Step~4 is one EFM.}
\label{fig:fig}
\end{figure}



\section{Results and Discussion}\label{sec:res}
\subsection{Computational Complexity} 
To calculate the computational complexity of Algorithm 2, the complexity of each step is analyzed below:
\begin{itemize}
  \item In Step 1, the maximum number of created pathways is calculated as
$P_{\rm{max}} = M_B  \times (r_{\rm{max}} )^{M_{\backslash B} }$, in which $M_B$ is the number of boundary metabolites, $M_{\backslash B}$ is the number of all other internal metabolites except the boundary ones and $r_{\rm {max}}$ is the maximum number of input/output reactions for all nodes.

  \item In Step 2, for each pathway $\mathcal{P}_i$, there are two traversals with the complexity of (1) $M_{\mathcal{P}_i}$, the number of all nodes in a pathway, and (2) $M_{t\mathcal{P}_i}$, the number of all tagged nodes with more than one input and one output in that pathway. Therefore, the order of maximum traversals is $P \times [M_{x\mathcal{P}_i}\times M_{\mathcal{P}_i}]$ or $P\times M^2$ where $M_{\mathcal{P}_i}$ and $M_{t\mathcal{P}_i}$ $\leq$ $M$.

  \item In Step 3, all nodes in $\mathcal{P}_i$ should be traversed which requires ${P^{\prime}} \times M_{\mathcal{P}_i}$ iterations. ${P^{\prime}} \leq P$ because some pathways may be discarded in Step 2. Therefore, the worst-case computational complexity for this step would be $P \times M$.

  \item In Step 4, for each $\mathcal{P}_i$ and for all nodes with secondary reactions, (i.e., $M_{v\mathcal{P}_i}$ nodes), at least half of the nodes should be traversed to lead the extra production/consumtion of metabolites to input/output. Therefore, ${P^{\prime}}\times[M_{v\mathcal{P}_i} \times M_{\mathcal{P}_i/2}]$ is the computational complexity of this step. Since $P^{\prime} \leq P$ and $M_{v\mathcal{P}_i},M_{\mathcal{P}_i/2} \leq M$, $P \times [M^2/2]$ is the computational complexity in the worst-case.
\end{itemize}

Considering all items above, $P \times O(M^2)$, or as a result ${r_{\rm{max}}}^M \times O(M^2)$ is the computational complexity of the algorithm if we assume $M_B \ll M$. The computational complexity of the approach can be considered as an upper-bound for the number of EFMs including external reactions.


\subsection{EFM Topology Description}
Three different topologies for an EFM can be observed in metabolic networks:
\begin{itemize}
  \item Internal pathways with no boundary reaction(s),
  \item pathways with only input or only output reaction(s) consisting of an internal loop,
  \item pathways from input reaction(s) to output reaction(s).
\end{itemize}
These topologies are illustrated in Figure~\ref{fig:EFMtopol}.
The main focus of Algorithm 2 is the third topology. However, in order to expand this focus to all EFMs, the algorithm can be easily modified to keep pathways without the contribution of output reactions (in Case 2) rather than discarding them, or starting from internal reactions (in Case 1).

However, from the biological point of view, internal pathways with no boundary reactions, (i.e., loops) are called ``futile cycles'' and are not biotechnologically relevant~\cite{burgard2001minimal}, and the pathways with boundary inputs and no output implicate inconsistency in the network~\cite{gevorgyan2008detection}. Therefore, a meaningful subset of EFMs are targeted here to keep more relevant EFMs in the output.

\begin{figure}[!tpb]
\centering
\includegraphics[scale=0.55]{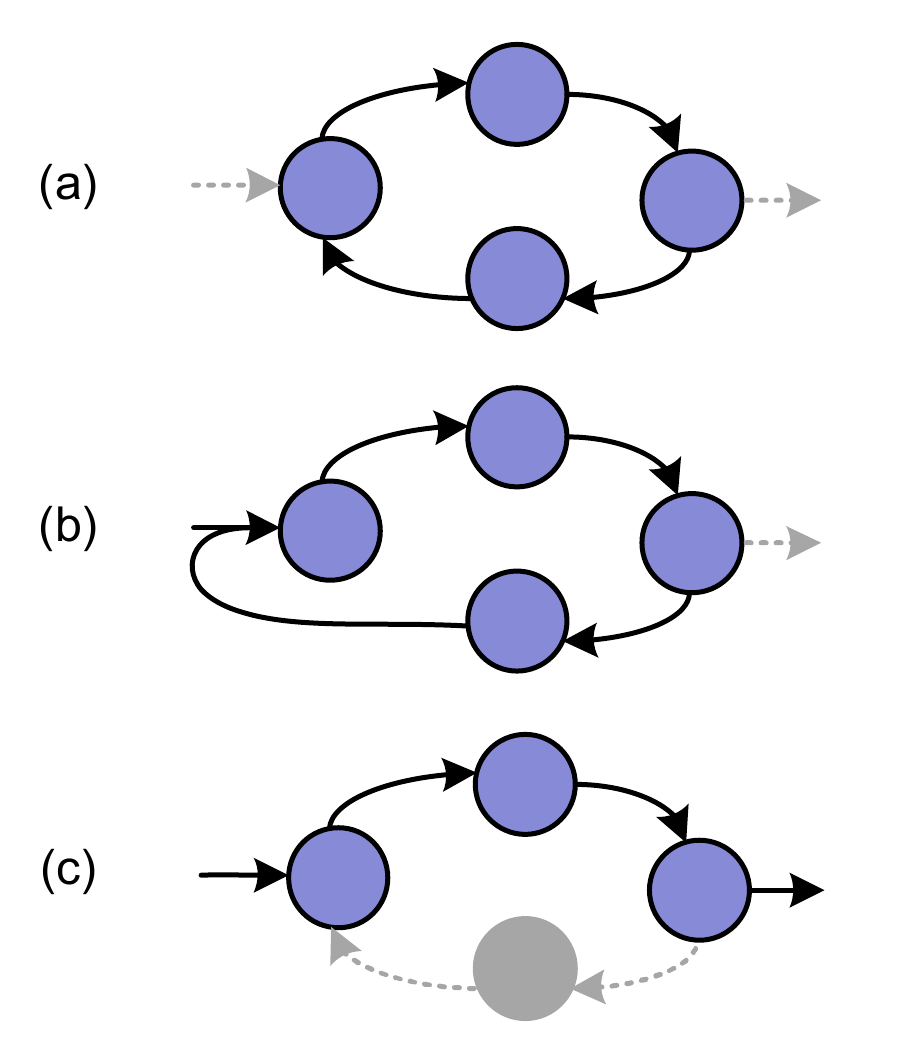}
\caption{Categorization of EFM topologies as explained in Section 4.2. An example of: (a) a ``futile cycle'' EFM, (b) an inconsistent route: a path with only input reaction consisting of an internal loop, and (c) a consistent non-cyclic EFM: a straight path (connecting external metabolites).}
\label{fig:EFMtopol}
\end{figure}

\subsection{Application to Metabolic Network Models}
To demonstrate the functionality of the model, the GB-EFM method has been implemented in C++ and tested on an Intel Core-i5 CPU with 4 GB RAM. The source code and results are provided on https://github.com/marabzadeh/GB-EFM.
A small metabolic network comprising tricarboxylic acid cycle, glyoxylate shunt and adjacent reactions of amino acid synthesis in \emph{E.} \emph{coli}~\cite{schuster1999detection} is considered to describe the steps of the algorithm and proof the concept. The network has one boundary input reaction and four boundary output reactions among all 25 reactions and 18 internal metabolites as shown in Table~1. Reactions $R_{7}$, $R_{12}$, $R_{13}$ and $R_{24}$ are boundary output reactions. The elementary flux modes were obtained using EFMtool~\cite{terzer2008large} as well as the proposed method.
In the case of considering all reactions as irreversible, EFMtool produces 12 EFMs in which 3 are cycles, 2 are inconsistent routes with only an input reaction and no output (mentioned in~\cite{gevorgyan2008detection} as elementary leakage modes) and the other 7 are consistent, non-cyclic pathways. Considering the reversible reactions in the table as reversible, 16 EFMs are produced, in which 3 and 2 are cycles and inconsistent routes belonging to the first and the second category as defined in Section 4.2, respectively, and 11 are pathways from input to an output as produced by GB-EFM. All information are provided in Supplementary file 1. The information of inconsistent routes with further description are provided in Supplementary file 2.

The number of pathways produced by the proposed algorithm after Step~1 is 40. After Step~2, this is reduced to 30 and after applying Steps~3-4, 5 pathways (final EFMs) remain. These pathways are the same as the eleven ``acyclic'' EFMs obtained by EFMtool. The resulting EFMs are illustrated in Figure~\ref{fig:EFMExample}.

\begin{table}
\caption{The list of reactions for the sample network of $E.$ $coli$ model and the resulting EFMs obtained by GB-EFM. External output metabolites and cofactors are marked by asterisks. The external metabolite consumed by the input boundary reaction $R_{15}$ is marked by two asterisks.}
\label{tab:table} {
\centerline{
\tiny
\begin{tabular}{lllcl|ccccccccccc}
\hline
Reactions  &&&&				&\multicolumn{11}{l}{EFMs} 	\\
\hline
\hline
\#	&Name   &Cons	&Dir 	&Prod	&1	&2	&3	&4	&5 &6 &7 &8 &9 &10 &11	\\
\hline
R0 	&Pyk 		 &PEP + ADP$^*$			&$\Rightarrow$ 		&Pyr + ATP$^*$       			&0	&1	&1	&3	&2	&3	&1	&2	&0	&2	&0\\	
R1 	&AceEF 		 &Pyr + NAD$^*$ + CoA 		&$\Rightarrow$ 		&AcCoA + CO$_{2}$$^*$ + NADH$^*$	&0	&0	&1	&3	&2	&3	&1	&2	&0	&2	&0\\
R2 	&GltA 		 &OAA + AcCoA 			&$\Rightarrow$ 		&Cit + CoA				&0	&0	&1	&2	&1	&2	&1	&1	&0	&1	&0\\
R3 	&Icd 		 &IsoCit + NADP$^*$ 		&$\Rightarrow$ 		&OG + CO$_{2}$$^*$ + NADPH$^*$		&0	&0	&1	&1	&0	&1	&1	&0	&0	&0	&0\\
R4 	&SucAB 		 &OG + NAD$^*$ + CoA 		&$\Rightarrow$ 		&SucCoA + CO$_{2}$$^*$ + NADH$^*$ 	&0	&0	&0	&0	&0	&1	&1	&0	&0	&0	&0\\	
R5 	&Icl 		 &IsoCit 			&$\Rightarrow$ 		&Succ + Gly				&0	&0	&0	&1	&1	&1	&0	&1	&0	&1	&0\\
R6 	&Mas 		 &Gly + AcCoA 			&$\Rightarrow$ 		&Mal + CoA				&0	&0	&0	&1	&1	&1	&0	&1	&0	&1	&0\\
R7 	&AspCon    	 &Asp 				&$\Rightarrow$ 		&Asp\_ex$^*$				&1	&0	&0	&0	&1	&0	&0	&0	&0	&0	&0\\
R8 	&AspA 		 &Asp 				&$\Rightarrow$ 		&Fum + NH3$^*$				&0	&0	&0	&0	&0	&0	&0	&0	&1	&0	&0\\	
R9 	&Pck 		 &OAA + ATP$^*$ 		&$\Rightarrow$ 		&PEP + ADP$^*$ + CO$_{2}$$^*$		&0	&0	&0	&0	&0	&0	&0	&0	&0	&0	&0\\
R10 	&Ppc 		 &PEP + CO$_{2}$$^*$ 		&$\Rightarrow$ 		&OAA					&1	&0	&1	&0	&0	&0	&1	&0	&1	&1	&1\\
R11 	&Pps 		 &Pyr + ATP$^*$ 		&$\Rightarrow$ 		&PEP + AMP$^*$				&0	&0	&0	&0	&0	&0	&0	&0	&0	&0	&0\\
R12 	&GluCon 	 &Glu 				&$\Rightarrow$ 		&Glu\_ex$^*$				&0	&0	&1	&1	&0	&0	&0	&0	&0	&0	&0\\	
R13 	&AlaCon 	 &Ala 				&$\Rightarrow$ 		&Ala\_ex$^*$				&0	&1	&0	&0	&0	&0	&0	&0	&0	&0	&0\\
R14 	&SucCoACon	 &SucCoA 			&$\Rightarrow$ 		&Suc\_ex$^*$ + CoA			&0	&0	&0	&0	&0	&1	&1	&1	&1	&2	&1\\
R15 	&Eno 		 &PG$^{**}$ 			&$\Leftrightarrow$ 	&PEP					&1	&1	&2	&3	&2	&3	&2	&2	&1	&3	&1\\
R16 	&Acn 		 &Cit 				&$\Leftrightarrow$ 	&IsoCit					&0	&0	&1	&2	&1	&2	&1	&1	&0	&1	&0\\
R17 	&SucCD 		 &SucCoA + ADP$^*$ 		&$\Leftrightarrow$ 	&Succ + ATP$^*$ + CoA			&0	&0	&0	&0	&0	&0	&0	&-1	&-1	&-2	&-1\\
R18 	&Sdh 		 &Succ + FAD$^*$ 		&$\Leftrightarrow$ 	&Fum + FADH$_{2}$$^*$			&0	&0	&0	&1	&1	&1	&0	&0	&-1	&-1	&-1\\	
R19 	&Fum 		 &Fum 				&$\Leftrightarrow$ 	&Mal					&0	&0	&0	&1	&1	&1	&0	&0	&0	&-1	&-1\\
R20 	&Mdh 		 &Mal + NAD$^*$ 		&$\Leftrightarrow$ 	&OAA + NADH$^*$				&0	&0	&0	&2	&2	&2	&0	&1	&0	&0	&-1\\
R21 	&AspC 		 &OAA + Glu 			&$\Leftrightarrow$ 	&Asp + OG				&1	&0	&0	&0	&1	&0	&0	&0	&1	&0	&0\\
R22 	&Gdh 		 &OG + NH$_{3}$$^*$ + NADPH$^*$ &$\Leftrightarrow$ 	&Glu + NADP$^*$				&1	&1	&1	&1	&1	&0	&0	&0	&1	&0	&0\\	
R23 	&IlvEAvtA 	 &Pyr + Glu 			&$\Leftrightarrow$ 	&Ala + OG  				&0	&1	&0	&0	&0	&0	&0	&0	&0	&0	&0\\
R24 	&SucEx		 &Suc 				&$\Leftrightarrow$ 	&Suc\_ex$^*$	  			&0	&0	&0	&0	&0	&1	&1	&1	&1	&2	&1\\
\hline
\end{tabular}}
}
\end{table}

\begin{figure}[!tpb]
\centerline{
\includegraphics[scale=0.1]{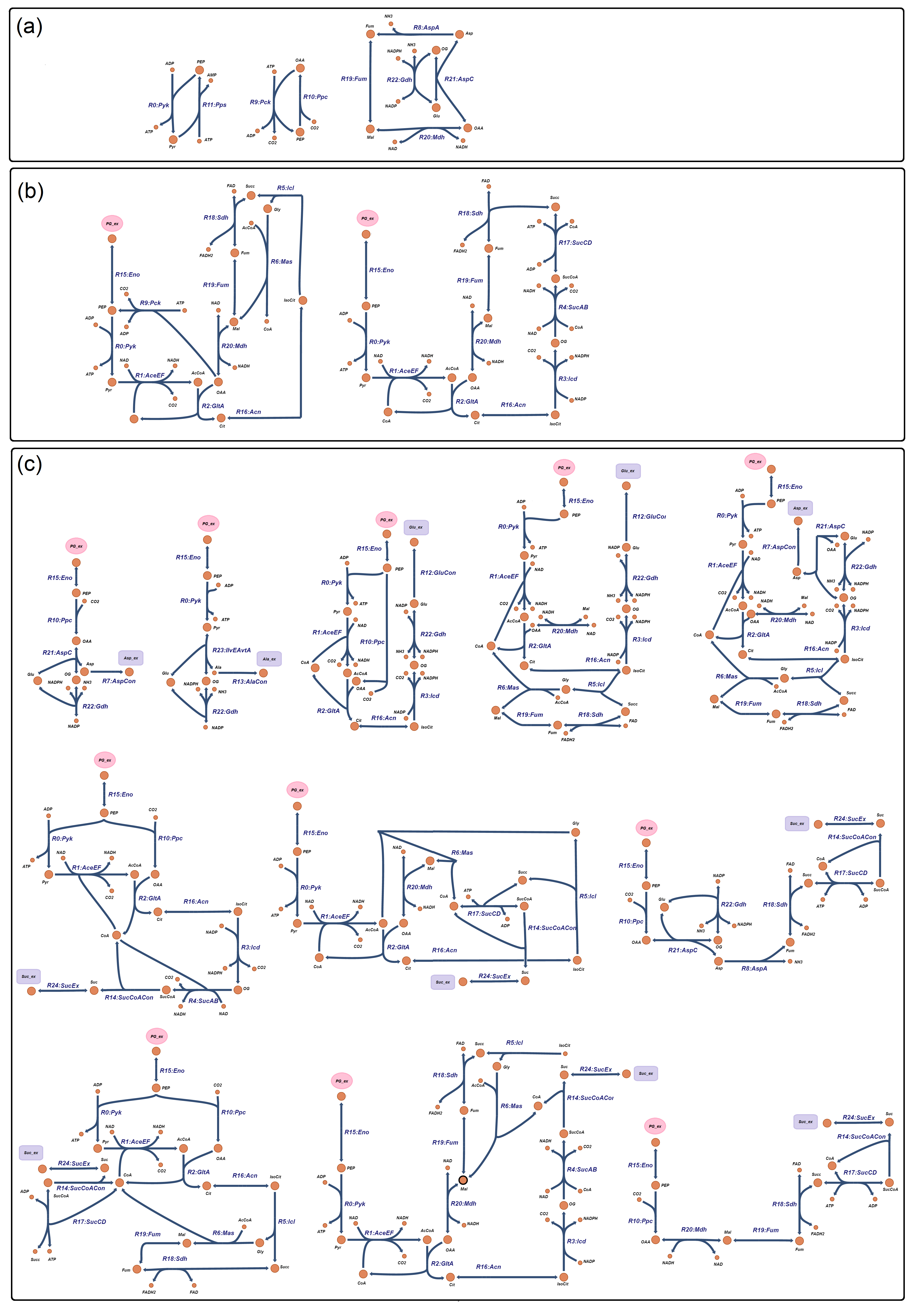}}
\caption{The resulting EFMs of the metabolic network comprising tricarboxylic acid cycle, glyoxylate shunt and adjacent reactions of amino acid synthesis in \emph{E.} \emph{coli}~\cite{schuster1999detection}. Abbreviations are explained in~\cite{schuster1999detection}. EFMs are depicted using Escher~\cite{king2015escher}. (a) ``futile cycle'' EFMs, (b) inconsistent routes: paths with only input reaction consisting of internal loops, with no output reaction, and (c) consistent non-cyclic EFMs: paths connecting external metabolites; produced by both EFMtool and GB-EFM. PG is an input metabolite shown by pink ellipse and Asp$\_$ex, Glu$\_$ex, Ala$\_$ex and Suc$\_$ex are output metabolites shown by purple squares. External and internal metabolites are illustrated by larger circles and currency metabolites by smaller ones.}
\label{fig:EFMExample}
\end{figure}

Besides, a set of moderate-sized metabolic networks are considered as further test cases. The first network is the pentose phosphate pathway in $trypanosoma$ $brucei$~\cite{kerkhoven2013handling}, the second one is the simple network of human red blood cell metabolism~\cite{wiback2002extreme} and the third one is an $E.$~$coli$ core model obtained from~\cite{king2016bigg}.
GB-EFM found all EFMs from \emph{Glucose} as reported in Table~\ref{tab:table2}. Computational time for all the networks was less than one second. The currency metabolites such as ADP, ATP, AMP, CO$_{2}$, H$_{2}$O, O$_{2}$, H$_{2}$, NH$_{4}$, Pi, NAD, NADH, NADP and NADPH were removed from the input stoichiometry matrix of all samples. As stated in the literature~\cite{kaleta2009efmevolver,de2009computing}, some simplifications to the system\textquotesingle s model can be performed in order to reduce the complexity of the problem, such as setting currency metabolites like cofactors and energy metabolites to external. According to ~\cite{kaleta2009efmevolver}, for energy currency metabolites like ATP, NADH and FADH, since their concentration is assumed to be constant they are not required to be balanced by an EFM. PG is an input metabolite and Asp$\_$ex, Glu$\_$ex, Ala$\_$ex and Suc$\_$ex are output metabolites.
\begin{table}
\caption{Number of EFMs starting from \emph{Glucose} in some networks.}
\label{tab:table2} {
\small
\begin{tabular}{lcc}
\hline
Network  							&Size ($m*n$)	 	&\#EFMs	\\
\hline
\hline
Glycolysis and Pentose Phosphate pathway in $T.$ $brucei$		&26$*$35						&4	\\
Human red blood cell						&20$*$50						&20	\\
$E.$ $coli$ core (Escherichia coli str. K-12 substr. MG1655)		&53$*$94						&47	\\	
\hline
\end{tabular}}
\end{table}

\subsection{Discussion}

According to the graph data model, the order of traversing the nodes and the relations between reactions and metabolites can be maintained for each pathway. Besides, decisions to select or eliminate a particular reaction or metabolite can be made during the algorithm. Since in this approach the traversal is on the graph, many pathways are produced which may be discarded afterwards.

The proposed data model gives the opportunity of different levels of parallelism to trace pathways. Therefore, the complexity of implementation can be reduced by optimizing the method to only focus on pathways that are desirable in EFM analysis and by exploiting advantages of the possible parallelism in the method. The graph structure makes a good track of unbalanced metabolites. Our approach tries to decide if a certain ``pathway'' is EFM or not, merely based on topology rules (instead of recognizing the elementarity through rank test or comparing new candidates with produced ones).

Our strategy brings up the opportunity of object-oriented programming by considering a metabolite as an object. The method at least in its present form, might not be applicable to large-scale networks. However, the graph structure makes it easier to set-up rules for pathways and explore the intended solution space via rules. The proposed method can potentially be implemented through a system design approach on hardware more conveniently.

A hypergraph is a simplified type of the modified AND/OR graph. The application of hypergraphs in biological networks has been reviewed in ~\cite{klamt2009hypergraphs}. In the introduced model, each arc has two different input and output coefficients. Besides, for each arc, a dynamic label for each pathway is defined, that is, the flux of the reaction represented by that arc. The proposed framework can be considered as an extension to the well-known flow network problem for hypergraphs. The edges are directed and each has been associated with a weight. However, the weight of an output arc of a node is different from the weight when the arc enters another node as an input. Besides, the hypergraph structure complicates pathways topology~\cite{marashi2014mathematical}. The complication arises since fluxes of contributing input and output reactions over a node may be related. This results to the non-linear property of fluxes. The non-linearity has been considered in our flux calculation procedure (see Eq. 2 and Eq. 3).

The first introduced double-description method for finding EFMs explores the whole set of pathways to find pathways that are in the steady-state solution space and then selects EFMs by comparing the reaction subset of pathways~\cite{schuster1994elementary}. In the improved double-description versions of the method which use null-space of the stoichiometry matrix as an input, different combinations of pathways in the null-space are calculated and then the reaction subset of pathways are compared~\cite{wagner2004nullspace,urbanczik2005improved,quek2014depth}. The main difference of the GB-EFM method with double-description-based methods is that GB-EFM first calculates the shortest pathways in the solution space and then checks to see if these pathways can be  in the steady-state condition by using some rules on the topology of the pathway.

The method of~\cite{cespedes2015new} explores the shortest pathways in the AND/OR graph of the network. The authors did not take into account the stoichiometry coefficients. Thus, the resulting pathways may or may not be EFM.  In the method of~\cite{ullah2016gefm}, all reaction combinations to balance a metabolite are calculated and all EFMs are computed accordingly. The order of traversing unbalanced metabolites is based on the graph model. The proposed model in this paper traverses the graph and keeps both the stoichiometry information and the minimality of the pathway with the opportunity of not exploring the whole solution space. The methods based on linear programming usually restrict their method to integer variables~\cite{de2009computing} which is limiting. In the case of the proposed method this is not a serious concern.

The presented approach, as discussed earlier, aims at finding EFMs by specific characteristics, i.e., the third category is Section 4.2. This brings up the fact that we are looking for new ways to reduce the solution space. EFMtool and the method of~\cite{ullah2016gefm} target the whole set of EFMs, which means either they find the complete set of EFMs of a metabolic network (which is a lot of non-categorized output information) or they fail to find any EFM. In comparison, our method is looking for a meaningful subset of EFMs which leads to a categorized output set with the cost of losing speed.

In several approaches such as~\cite{terzer2008large,de2009computing}, reversible reactions are considered as two irreversible reactions and the two-cycle EFMs are being removed during the EFM calculation. As it has been studied in~\cite{larhlimi2008inner}, by splitting reversible reactions, the size of the flux cone description increases by one. However, in the proposed method, while reversible reactions are added to both input and output of a metabolite, they are treated as reversible routes. When a reversible route is used in one direction, the other direction is disabled. Consequently, the proposed method preserves the dimension of the flux cone.

\section{Conclusion and Future Work}\label{sec:conc}
In order to calculate the elementary flux modes of a given metabolic network, an algorithm was proposed based on a modified AND/OR graph. It was shown in the paper that the steps in the algorithm leads to exact pathways according to EFM definition. The worst-case computational complexity of the algorithm was calculated as $O({r_{\rm{max}}}^M \times M^2)$ which is a newly reported upper-bound for the number of EFMs with external metabolites.
Additionally, a set of test cases was provided to prove the concept of the algorithm.

The main focus of this paper was to introduce a model to find minimal flux-balanced pathways in metabolic networks. The model can also be applicable to several constraint-based pathway analysis approaches. Using the potential parallelism in the method to speed-up the algorithm on a hardware structure and proposing different categorization of EFMs based on reactions/metabolites characteristics, as is possible according to the graph data model, are considered as our future research.

\section*{Acknowledgements}
Authors would like to thank Dr. Nathan Lewis (UCSD) for helpful discussions on the concept.



\end{document}